%% file: main.tex
\documentclass{style/vldb_style_sample/vldb}

\usepackage{graphicx}
\usepackage{balance}  

\usepackage{booktabs}
\usepackage{amsmath}
\usepackage{verbatim}
\usepackage{algorithmicx}
\usepackage[plain]{algorithm}
\usepackage{algpseudocode}
\usepackage{pgfplots}
\usepackage{pgfplotstable}
\usepackage[center]{subfigure}
\usepackage{url}

\usetikzlibrary{patterns,shapes,arrows,positioning,fit,decorations.pathreplacing}

\usepgfplotslibrary{groupplots}

\newlength{\abovecaptionskip}
\newfont{\mycrnotice}{ptmr8t at 7pt}
\newfont{\myconfname}{ptmri8t at 7pt}
\let\crnotice\mycrnotice%
\let\confname\myconfname%

\usepackage{xparse}
\let\oldState\State
\RenewDocumentCommand{\State}{o}{
  \IfValueTF{#1}{\makeatletter\setcounter{ALG@line}{#1}\addtocounter{ALG@line}{-1}\makeatother}{}
  \oldState\ignorespaces%
}

\begin{document}

\newtheorem{lemma}{Lemma}
\newtheorem{theorem}{Theorem}
\newtheorem{corollary}{Corollary}
\newtheorem{example}{Example}
\newtheorem{assumption}{Assumption}
\newtheorem{remark}{Remark}

\newdef{definition}{Definition}
\newdef{scenario}{Scenario}


\title{Probably Approximately Optimal Query Optimization\thanks{This work was supported by ERC Grant 279804 and by a European Google PhD fellowship.}}



%
%
%
%

\numberofauthors{1} 

\author{
%
%
\alignauthor
Immanuel Trummer and Christoph Koch\\
			 \email{\{firstname\}.\{lastname\}@epfl.ch}\\
       \affaddr{\'Ecole Polytechnique F\'ed\'erale de Lausanne}
}

\maketitle

\begin{abstract}
Evaluating query predicates on data samples is the only way to estimate their selectivity in certain scenarios. Finding a guaranteed optimal query plan is not a reasonable optimization goal in those cases as it might require an infinite number of samples. We therefore introduce probably approximately optimal query optimization (PAO) where the goal is to find a query plan whose cost is near-optimal with a certain probability. We will justify why PAO is a suitable formalism to model scenarios in which predicate sampling and optimization need to be interleaved.

We present the first algorithm for PAO. Our algorithm is non-intrusive and uses standard query optimizers and sampling components as sub-functions. It is generic and can be applied to a wide range of scenarios. Our algorithm is iterative and calculates in each iteration a query plan together with a region in the selectivity space where the plan has near-optimal cost. It determines the confidence that the true selectivity values fall within the aforementioned region and chooses the next samples to take based on the current state if the confidence does not reach the threshold specified as problem input. We devise different algorithm variants and analyze their complexity. We experimentally compare them in terms of the number of optimizer invocations, samples, and iterations over many different query classes. 
\end{abstract}

\input{sections/intro.tex}

\input{sections/related.tex}
\input{sections/model.tex}
\input{sections/algorithm.tex}
\input{sections/analysis.tex}
\input{sections/experimental.tex}
\input{sections/conclusion.tex}


\begin{tiny}
\bibliographystyle{abbrv}

\end{tiny}



\end{document}

%% file: sections/intro.tex
\section{Introduction}
\label{introSec}

The goal in traditional query optimization is to find an optimal query plan. This optimization goal is based on the assumption that accurate values for the selectivity of each query predicate are available. This assumption is already optimistic in traditional scenarios where predicate selectivity is estimated based on a machine-readable definition of the predicate and on statistics describing the input data. It is completely unjustified if query predicates are formulated in natural language (e.g., in the context of crowd databases where predicates are evaluated by human workers~\cite{Franklin2011,  Marcus2012a, Park2013a}), if predicates are evaluated by invoking remote services (e.g., a weather service or a credit rating service~\cite{Garcia-molina2014}), or if predicates are complex functions defined using a different language than SQL (e.g., predicates whose evaluation requires sophisticated text or video processing~\cite{Garcia-molina2014}). In all those cases, the optimizer cannot access or interpret the definition of query predicates and has a-priori no information about their selectivity.

In such cases it is still possible to evaluate predicates on data samples during optimization to estimate their selectivity. This is in fact required to make any guarantees on the quality of the generated plans. The traditional optimization goal is however not reasonable in this context: when taking it seriously, we would need to collect an unbounded number of samples in order to obtain infinitely precise selectivity estimates based on which we can produce a guaranteed optimal plan (as different plans might be optimal for different selectivity values). How should we formulate the optimization goal if optimization and selectivity sampling are interleaved? Which predicates shall we sample? How many samples should we take? How should we choose a query plan based on the collected samples? In the following we will first propose a generalization of the traditional query optimization problem and justify why it constitutes a reasonable optimization goal in our scenario. The answers to all other questions will follow from that. 

Our adapted optimization goal is the following: find a query plan whose cost is not higher than optimal by more than factor $\alpha\geq 1$ with probability at least $\delta\in[0,1]$. Factor $\alpha$ and probability $\delta$ configure the optimizer and can be chosen by users or administrators, trading plan quality for optimization time and sampling overhead. We call this optimization problem probably approximately optimal query optimization (similar in spirit to the probably approximately correct learning framework~\cite{Valiant1984}). It is a reasonable problem definition, meaning that it does not require an infinite number of samples. Simplified problem definitions do not possess that property as we show next.

Selectivity is a continuous variable that takes values between 0 and 1. Without making further assumptions, taking a finite number of samples can only yield confidence bounds for the true selectivity values but never a point estimate. The cost of query plans depends in general on the selectivity and may vary for selectivity values within the confidence bounds. Finding a plan that is optimal for all selectivity values within the confidence bounds is therefore in general not possible. We assume in the following that plan cost functions are continuous in a small area around the true selectivity values (additional assumptions are discussed in Section~\ref{modelSec}). This assumption is less strong than the standard assumption in parametric query optimization, called the guiding principle of parametric query optimization, stating that the cost of query plans is approximately linear in small areas of the selectivity space~\cite{Dey2008}. Under this assumption, we can shrink the selectivity confidence bounds using a finite number of samples until one plan has near-optimal cost (within multiplicative factor $\alpha$ of the optimum) within the entire cube in the selectivity space that is defined by the confidence bounds. This is why our goal must be to find a near-optimal and not an optimal plan. 

In addition, we must consider the case that the true selectivity values fall outside the confidence bounds altogether. As long as we take a finite number of samples, this is always possible with non-zero probability. This is the reason that we aim at finding a near-optimal plan with probability at least $\delta$ and not to find a guaranteed near-optimal plan. 


Prior work on selectivity estimation in the context of query optimization has often focused on the case of correlated predicates as they can lead to sub-optimal plan choices~\cite{Chaudhuri2001, Stillger2001}. In contrast to that, we focus on novel application scenarios where even the handling of uncorrelated predicates is currently not solved in a principled manner. A recently proposed system for big data processing~\cite{Karanasos2014a} for instance estimates the selectivity of user-defined predicates by sampling. The amount of data to sample is however decided heuristically such that no probabilistic guarantees on the optimality of the generated query plans result. In the context of crowd database systems, where users may define natural language predicates that are evaluated by crowd workers, query optimizers currently rely on the user to specify the predicate selectivity by hand~\cite{Park2013a}. Otherwise default values are used~\cite{Park2013}. Both scenarios can benefit from the framework and approach that we introduce in the following.

Prior work (e.g., \cite{Chaudhuri2001}) usually assumes that predicate selectivity is estimated based on data statistics which are generated in one atomic operation. This motivates a binary uncertainty model: either the selectivity of a predicate is known or unknown. With such a problem model, the goal of finding near-optimal plans is actually reasonable. We target scenarios where it is often not possible to estimate selectivity based on data statistics (e.g., how to automatically estimate selectivity for a predicate formulated in natural language for evaluation by crowd workers?~\cite{Park2013}). We can only sample predicates and never know the precise selectivity. Hence we cannot adopt a binary uncertainty model. Another distinguishing factor between our approach and prior work is that our approach requires a minimal amount of information. We do for instance not assume that table cardinalities are known which makes it possible to support even scenarios related to crowdsourcing where tables are extracted from crowd workers and their cardinality is not known. Those seemingly subtle differences result in algorithmic challenges that are specific to our scenario. 

We present a first algorithm for probably approximately optimal query optimization (PAO). This algorithm is generic and works with arbitrary plan spaces and cost models as long as cost grows monotonically in the selectivity. We assume that nothing is known about the selectivity of predicates before sampling starts and we do not make any assumptions at all about the type of probability distribution that the selectivity values follow. Furthermore, our approach is non-intrusive and can easily be implemented on top of existing optimizers. Possible application scenarios for the algorithm presented in this paper include cost-based optimizers in crowd database systems (which currently assume that users specify selectivity values together with their queries~\cite{Park2013, Park2013a}) or optimization of queries including user-defined predicates (sampling is used in that context but without formal guarantees on plan quality~\cite{Karanasos2014a}).

The optimizer consists of three components in our scenario: a standard optimizer that finds optimal plans for given selectivity values, a sampling component that evaluates predicates on data samples, and a meta-algorithm that coordinates invocations of the sampling component and the standard optimizer in order to meet the probabilistic guarantees specified by PAO. We propose a meta-algorithm for PAO in this paper while we assume that a standard optimizer and a sampling component are provided.

We use a standard optimizer that requires selectivity estimates for invocation. Those estimates are derived from samples. On the other hand, we can exploit information about the optimal cost function, namely its sensitivity to the selectivity of specific predicates, to choose the predicates to sample and to optimize the number of samples to take. This cyclic dependency motivates an iterative algorithm: our algorithm interleaves sampling and optimization and executes several stages of sampling and optimization until the PAO guarantees are met. 

After having collected a small initial set of samples, our algorithm executes the following four steps in each iteration: First, it invokes the standard optimizer to obtain the best query plan for the current selectivity sample means. Second, it issues more calls to the standard optimizer in order to determine a cube-shaped region in the selectivity space for which the plan generated in the first step satisfies the near-optimality guarantees specified as problem input. Third, it uses a probabilistic model based on the Hoeffding inequality~\cite{Hoeffding63} to lower-bound the confidence that the true selectivity values fall within the cube-shaped near-optimality region calculated in the second step. If this confidence value exceeds $\delta$ then optimization ends and the current query plan is returned. If the confidence is not yet sufficiently high, our optimizer finally chooses the number of samples to take in the next iteration and the best distribution of samples over different query predicates.

Finding a cube-shaped near-optimality region (step two) is challenging: on the one hand, we want to find a region of maximal extent as this maximizes the confidence and allows to terminate optimization early. On the other hand, finding the region of maximal extent might require a prohibitive number of calls to the standard optimizer as we need to find the optimal cost value for many points in the selectivity space. We will present and experimentally compare different variants of the region finding algorithm. All of them only execute a bounded number of standard optimizer invocations per iteration (in our experiments we use an optimizer budget of ten calls per iteration for instance). Our most sophisticated algorithm variant iteratively extends the near-optimality cube and calculates the gradient for the cost function at the cube bounds after each extension: this allows to estimate the extent of the near-optimality cube before extension and therefore to carefully select the selectivity vectors for which the standard optimizer is invoked.

We also propose different policies for selecting the next samples to take. We present variants that offer formal guarantees on achieving near-optimal tradeoffs between the number of optimzier calls and the number of samples taken and we present policies that significantly reduce the number of required iterations by estimating the required number of samples based on an analysis of the plan cost function and the selectivity sample distribution obtained so far. We evaluate those more sophisticated policies in comparison with uniform sampling.

In summary, our original scientific contributions in this paper are the following:

\begin{itemize}
\item We introduce and formalize the probably approximately optimal query optimization problem (PAO).
\item We present several variations of a generic and non-intrusive algorithm for PAO that is applicable in many scenarios.
\item We formally analyze different algorithm variants and we experimentally compare them for different query sizes, join graph shapes, and number of predicates.
\end{itemize}

The remainder of this paper is organized as follows. We discuss related work in Section~\ref{relatedSec}. In Section~\ref{modelSec}, we introduce our formal problem model and discuss the fundamental assumptions it is based on. We present several variants of a PAO algorithm in Section~\ref{algorithmSec} and formally analyze it in Section~\ref{analysisSec}. We finally  compare different algorithm variants in an experimental evaluation in Section~\ref{experimentalSec}.

%% file: sections/related.tex
\section{Related Work}
\label{relatedSec}

We propose an iterative algorithm that interleaves query optimization and predicate sampling. This connects our approach to the magic number sensitivity analysis (MNSA) algorithm by Chaudhuri and Narasayya~\cite{Chaudhuri2001} but there are several important differences. First, MNSA assumes that information is collected by collecting data statistics from which predicate selectivity is inferred while we assume that predicates are sampled directly. Our probabilistic model presented in Section~\ref{confidenceSub} is specific to our scenario. Second, MNSA adopts a binary approach to uncertainty: either the selectivity of a predicate is completely unknown (within the range $\varepsilon$ and $1-\varepsilon$) or it is completely known. When sampling predicates, it can take prohibitively many samples to determine the true selectivity with negligible error. We therefore work with confidence bounds on  predicate selectivities and our optimization goal (finding plans that are near-optimal with high probability) is different than the one of Chaudhuri and Narasayya (finding plans that are guaranteed near-optimal). Many of the problems that our algorithm presented in Section~\ref{algorithmSec} solves (e.g., finding regions within the selectivity space where a given plan is near-optimal with few optimizer calls) are specific to our scenario. Finally, MNSA assumes that one data statistic is created in each iteration. We assume that many samples can be taken per iteration (this reduces the required number of optimizer calls and allows to exploit parallelism during sampling) and deciding how many samples to take per iteration and how to partition them over different predicates is another important sub-problem that we consider.

MNSA was later generalized~\cite{Bruno2002} but the proposed approach assumes that MNSA or a similar technique is applied in a first step to generate statistics on base tables and that hard bounds on predicate selectivity values (not probabilistic confidence bounds as in our case) can be inferred from those statistics. We cannot infer selectivity from data statistics in our case (e.g., no generic data statistic would allow to automatically estimate the selectivity of natural language predicates evaluated by crowd workers on a picture database~\cite{Marcus2011a}) and MNSA is not suitable for our scenario as discussed before.

Joklekar et al.~\cite{Garcia-molina2014} propose a sampling approach to estimate the selectivity of user-defined predicates (mainly for the case of one predicate while an extension is shortly discussed). Their focus is on how to avoid evaluation of user-defined predicates by approximation while our focus is on interleaved join ordering and sampling, assuming that each predicate is evaluated on each relevant tuple. Joklekar et al.\ and others~\cite{Chaudhuri2009a} obtain confidence bounds on predicate selectivities but it is a-priori unclear how selectivity bounds translate into plan cost bounds and this is why we interleave optimization and sampling in an iterative approach. 

Karanasos et al.~\cite{Karanasos2014a} sample the selectivity of user-defined predicates by executing pilot runs; an initial query plan is generated based on those pilot runs, the plan might be re-optimized once new information becomes available during execution. Unlike our approach, the work by Karanasos et al.\ does not provide any probabilistic guarantees on the near-optimality of the generated plans. Another difference is that we carefully choose the number of samples to take for each predicate in each iteration based on how sensitive the optimal join order is to the selectivity of different predicates. Re-optimization~\cite{Avnur2000, Babu2005, Graefe1989, Karanasos2014a} at run time is a generic technique to cope with uncertain predicate selectivities. However, an initial query plan with which to start execution must still be chosen and a sub-optimal initial plan can lead to huge overheads. Our approach can be used to generate that initial plan.

Seppi et al.~\cite{Seppi1993} propose a Bayesian approach to query optimization: based on prior probability distributions on the selectivity of predicates, the optimizer decides how much sampling is appropriate to choose between alternative operator implementations. Their approach relates to ours since the amount of samples to take is decided based on statistical models. The algorithm by Seppi et al.\ is however non-iterative and whether an appropriate number of samples is chosen depends entirely on the accuracy of the prior knowledge on predicate selectivities. In our scenario, nothing is known about the predicates a-priori. Further, the approach by Seppi et al.\ does not consider join order but only operator implementations. Finally, the probabilistic model by Seppi et al.\ directly operates on probability distributions over cost values; standard cost formulas need to be simplified in order to make the approach applicable~\cite{Seppi1993} and the used formulas must be tailored to specific scenarios. Our approach operates on probability distributions over selectivity values but not on probability distributions over cost values. Our approach is therefore applicable for arbitrary cost models (as long as cost values grow monotonically in the selectivity) and does not require adaption to specific scenarios. 

Optimization techniques that avoid selectivity estimation altogether~\cite{Dutt2014, Dutt2014a} incur huge overheads during planning and execution (optimization time is in the order of hours and execution time exceeds the optimal execution time in practice by factor 10 for some queries). The work by Goel et al.~\cite{Goel2006, Guha2007} is related to ours as it addresses the question of how to efficiently invest resources to reduce uncertainty at optimization time. However, the problem model used by Goel et al.\ does not match our scenario. It assumes that probability distributions are initially known for all parameters and that a subset of parameters is selected for which exact values can be obtained by probing. In our case, no probability distributions on the selectivity of predicates are initially known and exact selectivity values cannot be obtained except by evaluating each predicate on each tuple which is equivalent to query execution in our scenario. The algorithms proposed by Goel et al.\ only treat the case of simple cost functions that are defined as maximum or minimum over all parameters. Query optimization is cited as application scenario by Goel et al.\ but with the focus of selecting an optimal subset of queries to execute; the presented approaches cannot find optimal join orders.

Parametric query optimization~\cite{Cole1994, Ganguly1998, Hulgeri2002} assumes that predicate selectivities are unknown at optimization time. The goal is usually to generate a set of plans containing an optimal or near-optimal plan for each possible combination of selectivity values. In contrast, our goal is to find one plan that is near-optimal with a certain probability. Further, unlike in parametric query optimization, we assume that the optimizer can reduce uncertainty at optimization time by sampling predicates. This assumption motivates the design of our iterative algorithm. The goal in robust query optimization~\cite{Abhirama2010a, Babcock2005, D.2008} is to produce plans that are robust towards selectivity estimation errors, using a given probability distribution over selectivity values to judge robustness. It does not consider the possibility of taking samples during optimization which is central to our scenario. 

%% file: sections/model.tex
\section{Model and Assumptions}
\label{modelSec}

Unlike a typical publication on query optimization, we will neither introduce a specific query model nor a specific cost model nor a specific plan model in this section. The reason is that, strictly speaking, we do not propose an algorithm for query optimization but rather a meta-algorithm that uses an existing optimizer and sampling infrastructure to generate plans that satisfy certain probabilistic quality guarantees. 

A query describes data that needs to be generated. To make the following definitions more concrete, we can for instance imagine that queries are described in some extension of SQL that allows user-defined functions as predicates or predicates that are formulated in natural language and evaluated by crowd workers~\cite{Marcus2012a}. We make however only very generic assumptions on the query language and our algorithm is in principle applicable to other structured query languages as well. 

Queries are associated with predicates that need to be evaluated on data subsets. In case of SQL, predicates are defined on single base tables or between table subsets. Predicates are associated with a selectivity value between 0 and 1 which is the fraction of input data items that will satisfy the predicate. Estimating the selectivity of predicates is crucial as the choice of the optimal processing plan depends on them. Queries might contain predicates whose selectivity is known or can be reasoned about. We assume however that queries contain additionally predicates whose selectivity is unknown and can only be estimated by evaluating predicates on randomly selected data samples which we call short to \textit{sample the predicate}. For a given query $q$, we refer to the set of predicates with unknown selectivity values by $pred(q)$.

Many alternative query plans are in general available to process a given query. The execution cost of a query plan depends on the predicate selectivities. We use in the following the term \textit{selectivity vector} to refer to a vector containing a selectivity value for each predicate of a given query. Given two selectivity vectors $s_1$ and $s_2$, we write $s_1\preceq s_2$ to indicate that the selectivity value assigned by $s_1$ is smaller or equal to the selectivity assigned by $s_2$ for each predicate.

As selectivity values are unknown in our scenario, we must think of plan cost as a function that maps selectivity vectors to cost values. We assume that a cost model is available that maps plans and selectivity vectors to cost estimates. We represent this cost model by the function call \Call{Cost}{$p,s$} for a plan $p$ and a selectivity vector $s$ in our pseudo-code. If \Call{Cost}{$p,s$}$=c$ then we say that plan $p$ has cost $c$ at or for selectivity vector $s$.

The algorithm that we present in this paper is based on a generic assumption on the cost model: if $s_1\preceq s_2$ then the cost of any fixed plan is not higher for $s_1$ than for $s_2$, i.e.\ cost functions are monotone in the selectivity. This assumption is common in query optimization~\cite{Bizarro2009, Chaudhuri2001}. It is intuitive as higher predicate selectivity may lead to larger intermediate results and processing larger results is usually more costly. In the context of SQL, it holds for single-block SQL queries while it may not always hold for multi-block queries~\cite{Chaudhuri2001}.

For a fixed selectivity vector, we say that a plan is optimal if it has minimal cost among all possible query plans. We assume that a standard query optimizer is available that determines for a given selectivity value an optimal plan. We represent that optimizer by the function call \Call{BestPlan}{$q,s$} for query $q$ and selectivity vector $s$ in our pseudo-code. We say that a plan is $\alpha$-optimal (we also say \textit{near-optimal} as long as it is clear from the context to which value of $\alpha$ we refer) for a specific selectivity vector if its cost at $s$ is not higher than the optimal cost at $s$ by more than factor $\alpha$. 

We call the set of all possible selectivity vectors for a given query the selectivity space and a subset of vectors a region. A plan is $\alpha$-optimal in a selectivity space region if it is $\alpha$-optimal for each contained vector. We are particularly interested in (hyper)-cube shaped regions of the selectivity space that are defined by two selectivity vectors, a lower bound vector $lb$ and an upper bound vector $ub$ such that $lb\preceq ub$. The corresponding cube contains all vectors $s$ such that $lb\preceq s$ and $s\preceq ub$. In two dimensions, a cube corresponds for instance to a rectangle with the upper bound as right-upper corner and the lower bound as lower-left corner. We will associate plans with selectivity cubes in which they are near-optimal and that we call \textit{near-optimality cubes} in that context.

Our goal in this paper is to solve the probably approximately optimal query optimization problem: a problem instance is defined by a query $q$, an approximation factor $\alpha$, and a confidence threshold $\delta$. The goal is to find a query plan for $q$ that is $\alpha$-optimal with confidence $\delta$.

%% file: sections/algorithm.tex
\section{Algorithm}
\label{algorithmSec}

We present a first algorithm for PAO in this section. We discuss the main function in Section~\ref{mainAlgSub}. In Section~\ref{cubeSub}, we present different algorithm variants for the sub-function that identifies a cube-shaped region in the selectivity space for which a given plan has near-optimal cost. In Section~\ref{confidenceSub}, we discuss how to calculate the confidence that the true selectivity values fall within a given region in the selectivity space. In Section~\ref{samplingSub}, we present different policies for choosing the number of samples to take for each predicate in the next iteration.

\subsection{Main Algorithm}
\label{mainAlgSub}

\begin{algorithm}[t!]
\renewcommand{\algorithmiccomment}[1]{// #1}
\begin{footnotesize}
\begin{algorithmic}[1]
\State \Comment{Returns query plan for query $q$ that is $\alpha$-optimal}
\State \Comment{with probability at least $\delta$.}
\Function{PAOQ}{$q,\alpha,\delta$}
\State \Comment{Initialize selectivity sample set $S$}
\State $S\gets$\Call{InitSamples}{$q$}
\State \Comment{Until sufficiently confident to have $\alpha$-optimal plan}
\Repeat\label{mainLoopStart}
\State \Comment{Optimize for current selectivity means}
\State $p\gets$\Call{BestPlan}{$q,S.means$}
\State \Comment{Find selectivity cube where $p$ is $\alpha$-optimal}
\State $noc\gets$\Call{NearOptCube}{$q,p,\alpha,S$}
\State \Comment{Calculate confidence that selectivity in cube}
\State $conf\gets$\Call{Confidence}{$noc,S$}\label{mainCalcConfLine}
\State \Comment{Determine which samples to take next}
\State $nrSamples\gets$\Call{ChooseSamples}{$q,p,\alpha,\delta,S,noc$}
\State \Comment{Take corresponding number of samples}
\State $S\gets S\cup$\Call{TakeSamples}{$nrSamples$}
\Until{$conf\geq\delta$}\label{mainLoopEnd}
\State \Return{$p$}
\EndFunction
\end{algorithmic}
\end{footnotesize}
\caption{Interleaved sampling and optimization for probably approximately optimal query optimization.}
\label{paoAlg}
\end{algorithm}

Algorithm~\ref{paoAlg} is the main function of our algorithm. The input is a PAO problem specification, consisting of a query $q$ to optimize, an approximation factor $\alpha$ defining what constitutes a near-optimal cost value, and a confidence threshold $\delta$. The output is a query plan for query $q$ whose execution cost is not higher than optimal by more than factor $\alpha$ with probability at least $\delta$. We motivated this problem statement in Section~\ref{introSec}.

Our algorithm obtains very rudimentary initial selectivity estimates by evaluating each predicate on a small data sample (function~\textproc{InitSamples} in the pseudo-code). We can either use a small fixed number of evaluations per predicate. Or we can use an initialization method similar to the one used by Karanasos et al.~\cite{Karanasos2014a} and evaluate each predicate on its input tuples until we find the first tuple for which the predicate is satisfied. The initial sample will in general be insufficient to generate a plan that is near-optimal with sufficiently high confidence. The bigger part of samples is taken during the iterations of the main loop (lines~\ref{mainLoopStart} to \ref{mainLoopEnd}) that we explain next.

We execute the following steps in each iteration of the main loop. First, a query plan is created that is optimal when assuming that the selectivity value for each predicate is the arithmetic mean of all samples collected for that predicate so far. Function~\textproc{BestPlan} represents the invocation of a standard query optimizer that generates the best query plan for given selectivity estimates. Having an optimal plan for the selectivity means, we calculate a cube-shaped region in the selectivity space for which that plan has near-optimal cost (within factor $\alpha$ of the optimum). We use function~\textproc{NearOptCube} to calculate the near-optimality region. We discuss alternative implementations of that function in the next subsection. Having a selectivity region for which the current plan is near-optimal, we calculate the confidence that the true selectivity values fall into that region using function~\textproc{Confidence}. We discuss the formulas that we use to calculate confidence in Section~\ref{confidenceSub}. 

We finally choose the number of samples to take for each predicate using function~\textproc{ChooseSamples}. The output of \textproc{ChooseSamples} is a vector specifying for each query predicate the number of samples to take next. We discuss different policies for choosing the number of samples in Section~\ref{samplingSub}. The choice of the number of samples might be based for instance on the sensitivity of the cost function of the current plan towards the selectivity in different predicates. Having chosen the number of samples to take, we draw those samples using a sampling component represented by function~\textproc{TakeSamples}. Iterations end once the required confidence value is reached. 

Note that we can of course check whether the current confidence value is sufficient immediately after calculating confidence (line~\ref{mainCalcConfLine}) and avoid taking more samples in that case. The pseudo-code shown in Algorithm~\ref{paoAlg} is simplified while the implementation used for our experiments performs that check. Also note that we can introduce additional termination conditions, restricting for instance the number of iterations, standard optimizer calls, or the number of samples taken. In particular, we can introduce an additional termination condition that bounds the total number of samples based on the estimated execution cost of the current plan, thereby avoiding disproportional sampling overhead.

\subsection{Finding Near-Optimality Cubes}
\label{cubeSub}

We describe different possibilities to implement function \textproc{NearOptCube} in Algorithm~\ref{paoAlg}. Our goal is to find a selectivity cube in which a given plan has near-optimal cost. If possible, we want to find the cube with maximal extent among all cubes satisfying the aforementioned constraints. The reason is that we terminate the iterations of the main loop once our confidence that the true selectivity values lie within that cube reaches threshold $\delta$. Obtaining a bigger cube allows earlier termination. Verifying whether a plan has near-optimal cost in a given cube requires however calls to the standard optimizer. The number of optimizer calls required for finding a cube of maximal extent might be prohibitive. Hence we devise algorithms that try to find the largest possible near-optimality cube with a given budget of optimizer calls. Note also that the true region in which a given plan is optimal is in general not a cube. We still choose to approximate it by a cube since cubes are easy to represent and since the statistical model presented in Section~\ref{confidenceSub} is specific to cube-shaped regions.

\begin{algorithm}[t!]
\renewcommand{\algorithmiccomment}[1]{// #1}
\begin{algorithmic}[1]
\State \Comment{Extend upper and lower bounds of selectivity }
\State \Comment{space cube $noc$ by a fixed ratio.}
\Function{Extension}{$noc$}
\State $noc.lb\gets noc.lb\cdot (1-\epsilon)$
\State $noc.ub\gets noc.ub\cdot (1+\epsilon)$
\State \Comment{Make sure that selectivities are at most 1}
\State $noc.ub\gets\min(noc.ub,$\Call{Ones}{}$)$
\State \Return{$noc$}
\EndFunction
\vspace{0.1cm}
\State \Comment{Checks if plan $p$ for query $q$ is $\alpha$-optimal within}
\State \Comment{selectivity space cube $noc$.}
\Function{IsNearOptCube}{$q,p,\alpha,noc$}
\State $ubCost\gets$\Call{Cost}{$p,noc.ub$}
\State $lbPlan\gets$\Call{BestPlan}{$q,noc.lb$}
\State $lbCost\gets$\Call{Cost}{$lbPlan,noc.lb$}
\State \Return{$ubCost\leq \alpha\cdot lbCost$}
\EndFunction
\vspace{0.1cm}
\State \Comment{Returns selectivity space cube containing means of}
\State \Comment{samples $S$ in which plan $p$ is $\alpha$-optimal for query $q$.}
\Function{NearOptCube}{$q,p,\alpha,S$}
\State \Comment{Assign current sample means as initial bounds}
\State $ex.lb\gets S.means$
\State $ex.ub\gets S.means$
\State $nrOpt\gets 0$
\State \Comment{While $ex$ is $\alpha$-optimal and optimizer budget left}
\Repeat
\State \Comment{Current cube becomes last extension}
\State $noc\gets ex$
\State \Comment{Extend near-optimality cube}
\State $ex\gets$\Call{Extension}{$noc$}
\State \Comment{Count optimizer call for near-optimality check}
\State $nrOpt\gets nrOpt+1$
\Until{$\neg$\Call{IsNearOpt}{$q,p,\alpha,ex$}$\vee nrOpt\geq b$}
\State \Return{noc}
\EndFunction
\end{algorithmic}
\caption{Naive fixed-ratio algorithm for finding a near-optimality cube in the selectivity space.\label{cubeNaiveAlg}}
\end{algorithm}

We present two alternative algorithms for implementing function~\textproc{NearOptCube} in this subsection (Algorithm~\ref{cubeNaiveAlg} and \ref{cubeComplexAlg}). Algorithm~\ref{cubeNaiveAlg} is a rather naive approach at finding a large near-optimality cube with a given number of optimizer calls. The input is a plan $p$ for query $q$ together with the approximation factor $\alpha$ and the set of selectivity samples $S$ collected so far. The output is a near-optimality cube, $noc$, in which $p$ is $\alpha$-optimal.

We start with a cube of volume zero whose upper and lower bounds are the selectivity mean vectors from the current samples. This cube is iteratively extended until $p$ is not $\alpha$-optimal within the extended cube anymore or the number of optimizer calls exceeds the budget $b$. Cube extensions are realized by function~\textproc{Extension} by multiplying the upper cube bound vector by factor $(1+\epsilon)$ and the lower bound by factor $(1-\epsilon)$. The choice of $\epsilon$ is crucial and having a parameter to tune is a weakness of this first algorithm. In principle, we could also extend cubes by increasing or decreasing the bounds by an additive constant. The problem is however that adding a constant value to small selectivity values has large impact on the plan cost function while adding the same constant to a large selectivity values has little impact; it is therefore preferable to extend the bounds by a constant ratio. 

After each cube extension, function~\textproc{IsNearOptCube} verifies whether plan $p$ is still $\alpha$-optimal within the extended cube. This check requires one optimizer call and compares the optimal cost at the lower cube bound, $lbCost$, with the cost of the current plan $p$ at the upper cube bound, $ubCost$. The near-optimality check is based on the following theorem, similar results were exploited by Chaudhuri and Narasayya~\cite{Chaudhuri2001} and by Bizarro et al.~\cite{Bizarro2009}.

\begin{theorem}
If $ubCost\leq\alpha\cdot lbCost$ then plan $p$ is $\alpha$-optimal between the selectivity cube bounds.\label{nearOptTheorem}
\end{theorem}
\begin{proof}
The optimal cost $lbCost$ at the lower cube bound is at the same time a lower bound for the optimal cost in the entire cube as cost functions are assumed to be monotone in the selectivity. The cost of plan $p$ at the upper cube bound is an upper bound on the cost of plan $p$ in the entire cube for the same reason. Hence if $ubCost\leq\alpha\cdot lbCost$ then the cost of plan $p$ is not higher than the optimum by more than factor $\alpha$ in the entire cube. 
\end{proof}

\begin{algorithm}[t!]
\renewcommand{\algorithmiccomment}[1]{// #1}
\begin{algorithmic}[1]
\State \Comment{Calculates recommended step size when trying to}
\State \Comment{reach target cost $tcost$ for plan $p$ by making}
\State \Comment{a step from selectivity vector $sel$.}
\Function{StepSize}{$p,sel,tcost$}
\State \Comment{Calculate plan cost at selectivity vector}
\State $cost\gets$\Call{Cost}{$p,sel$}
\State \Comment{Calculate margin to target cost}
\State $margin\gets tcost - cost$
\State \Comment{Calculate cost slope at selectivity vector}
\State $slope\gets $\Call{Ones}{}$\cdot\nabla $\Call{Cost}{\textbf{fix} $p,X$}$|_{X=sel}$
\State \Comment{Return recommended step size}
\State \Return{$margin/slope$}
\EndFunction
\vspace{0.1cm}
\State \Comment{Extends cube $noc$ by adding $ubStep$ and $lbStep$}
\State \Comment{to upper and lower cube bounds respectively.}
\Function{Extension}{$noc,lbStep,ubStep$}
\State $noc.lb\gets noc.lb+lbStep\cdot$\Call{Ones}{}
\State $noc.ub\gets noc.ub+ubStep\cdot $\Call{Ones}{}
\State \Comment{Verify that selectivity within admissible range}
\State $noc.lb\gets\max(noc.lb,$\Call{Zeros}{}$)$
\State $noc.ub\gets\min(noc.ub,$\Call{Ones}{}$)$
\State \Return{$noc$}
\EndFunction
\vspace{0.1cm}
\State \Comment{Returns selectivity space cube containing means of}
\State \Comment{samples $S$ in which plan $p$ is $\alpha$-optimal for query $q$.}
\Function{NearOptCube}{$q,p,\alpha,S$}
\State \Comment{Assign current sample means as initial bounds}
\State $noc.lb\gets S.means$
\State $noc.ub\gets S.means$
\State \Comment{Determine target cost values for cube bounds}
\State $c\gets$\Call{Cost}{$p,S.means$}
\State $lbtc\gets c/\sqrt{\alpha}$
\State $ubtc\gets c\cdot\sqrt{\alpha}$
\State \Comment{Until optimizer budget runs out}
\State $nrOpt\gets 0$
\Repeat
\State \Comment{Calculate best plan for lower cube bounds}
\State $lbPlan\gets$\Call{BestPlan}{$q,noc.lb$}
\State $nrOpt\gets nrOpt+1$
\State \Comment{Calculate recommended step size at bounds}
\State $lbStep\gets$\Call{StepSize}{$lbPlan,noc.lb,lbtc$}
\State $ubStep\gets$\Call{StepSize}{$p,noc.ub,ubtc$}
\State \Comment{Reduce step size until extension possible}
\State $extended\gets$\textbf{false}
\Repeat
\State $ex\gets$\Call{Extension}{$noc,lbStep,ubStep$}
\State \Comment{Check if extended cube is near-optimal}
\If{\Call{IsNearOptCube}{$q,p,\alpha,ex$}}
\State $noc\gets ex$
\State $extended\gets$\textbf{true}
\Else
\State $lbStep\gets lbStep/2$
\State $ubStep\gets ubStep/2$
\EndIf
\State $nrOpt\gets nrOpt+1$
\Until{$extended\vee nrOpt\geq b$}
\Until{$nrOpt\geq b$}
\State \Return{noc}
\EndFunction
\end{algorithmic}
\caption{Finding near-optimality cubes with slope-based choice of extension step size.\label{cubeComplexAlg}}
\end{algorithm}

We can improve Algorithm~\ref{cubeNaiveAlg} in several ways. First, instead of terminating once a cube extension does not yield a near-optimality cube anymore, we can backtrack and try to extend the last cube with a smaller step size. Second, our goal should be to extend the upper and the lower bound by a similar amount since we will calculate confidence (that the true selectivity values falls into the cube) based on the minimal distance between selectivity mean and the cube bounds. Having extended either the upper or the lower bound by a big margin is useless if the other bound remains close to the mean as the confidence remains low. Third, extending the selectivity of each predicate by the same ratio is too inflexible as it neglects that the sensitivity of the cost function towards changes of the selectivity values might differ across different plans. Algorithm~\ref{cubeComplexAlg} integrates all of the aforementioned improvements.

The selectivity cube is initialized in the same way as before. Before starting the extensions, we calculate however a target cost value for the upper (variable $ubtc$) and the lower cube bound (variable $lbtc$). The semantics of those cost values is that our algorithm tries to extend the cube in a way such that the cost of plan $p$ at the upper cube bound corresponds approximately to $ubtc$ and the optimal cost value at the lower cube bound corresponds approximately to $lbtc$. This choice connects to the second improvement: if plan cost functions were linear then making the cost of $p$ at the upper bound equal to $ubts$ and the optimal cost at the lower cube bound equal to $lbtc$ would result in the near-optimality cube maximizing the minimal distance between selectivity mean and cube bounds. Note that we do not require plan cost functions to be linear and our algorithm works for non-linear cost functions as well. It will however tend to be more efficient the more closely the cost function resembles a linear function around the current selectivity means.

Each iteration of the outer loop in function~\textproc{NearOptCube} performs two steps. First, a recommended step size is calculated separately for the lower and the upper cube bound. Second, starting with the recommended step sizes we try to extend the current cube, cutting the step sizes by two until either the extended cube is small enough such that $p$ is $\alpha$-optimal within it or the optimizer budget runs out.

Function~\textproc{StepSize} calculates the recommended step size to start the extensions with. We choose the step size based on the sensitivity of the cost function which realizes the third of the three improvements we discussed before. We have calculated target cost values that we want to reach at the upper and lower cube bounds. The step size is based on the margin, i.e.\ on by how much the current cost value differs from the target. We also need to take into account how quickly the cost function changes in the selectivity parameters around the bounds. The expression \Call{Ones}{}$\cdot\nabla $\Call{Cost}{\textbf{fix} $p,X$}$|_{X=sel}$ calculates the slope of the cost function of plan $p$ in the direction of the unit vector (function~\textproc{Ones}) and the recommended step size takes the slope and the margin into account. In order to choose the step size at the lower cube bounds we examine the cost function of the plan that is optimal at the lower cube bounds (as we are interested in the slope of the optimal cost function). At the upper bound we examine the cost function of plan $p$ that is optimal for the current selectivity means. 

The calculated step size is only used as a recommendation and if extension is not possible with that step size then it is reduced. This is the purpose of the inner loop in function~\textproc{NearOptCube} and implements the first of the improvements proposed before. We compare both algorithms for finding near-optimality cubes in our experimental evaluation in Section~\ref{experimentalSec}.

\subsection{Calculating Confidence}
\label{confidenceSub}

Given a cube-shaped selectivity region in which a plan $p$ is $\alpha$-optimal, function~\textproc{Confidence} calculates the confidence that the true selectivity values fall within that region. This is at the same time our confidence that plan $p$ has near-optimal cost. We do not provide pseudo-code for function~\textproc{Confidence} but we describe the mathematical formulas that it applies.

We denote by $t$ the vector of true selectivity values in the following. We are interested in the probability that the true selectivity values fall within cube $noc$, i.e.\ we are interested in $\Pr(t\in noc)=\Pr(noc.lb\preceq t\preceq noc.ub)$. We want to lower-bound that probability and the resulting bound is our confidence value.

We denote by the subscript notation the selectivity value for one specific predicate (e.g., $t_r$ denotes the true selectivity of predicate $r$). The probability that the true selectivity is within the cube equals the probability that the selectivity of each predicate is within its corresponding cube bounds, i.e.\ $\Pr(t\in noc)=\prod_{r\in pred(q)}\Pr(noc.lb_r\leq t_r\leq noc.ub_r)$, denoting by $pred(q)$ the set of predicates in query $q$ that need to be sampled. Note that we assume independence between the deviations of the current mean from the true selectivity between different predicates.

The true selectivity for a predicate is at the same time the expected value for the mean of the corresponding selectivity samples. We can therefore upper-bound the probability that the true selectivity of one specific predicate falls outside the cube bounds by the probability that an expected value deviates from the current mean. The Hoeffding inequality~\cite{Hoeffding63} provides us with corresponding bounds while making no assumptions on the probability distribution of the sampling distribution. 

More precisely, denote by $a_r$ the arithmetic average value of the selectivity samples collected for one specific predicate $r$. Denote by $t_r$ its true selectivity  and denote by $l_r$ the number of samples (i.e., the number of tuples on which the predicate was evaluated) collected for that predicate so far. The cube bounds for $r$ are given by $noc.lb_r$ and $noc.ub_r$ and due to how we generate cubes, the mean value must be contained between the cube bounds: $noc.lb_r\leq a_r\leq noc.ub_r$. Denote by $d_r=\min(a_r-noc.lb_r,noc.ub_r-a_r)$ the minimal distance between the current mean and the cube bounds. Then the probability that the true selectivity of $r$ falls outside the cube bounds, $\Pr(t_r<noc.lb_r\vee t_r>noc.ub_r)$, is bounded by the probability that the distance between current mean and true selectivity exceeds $d_r$. 

As the true selectivity is the expected value of the mean, we can bound that probability using the Hoeffding inequality. Assuming that $noc.lb_r>0$ and $noc.ub_r<1$, we have to apply the two-sided version of the Hoeffding inequality (otherwise we can use the one-sided version which provides a tighter bound). In that case, we obtain $2\cdot\exp(-2\cdot d_r^2\cdot l_r)$ as upper bound for $\Pr(t_r<noc.lb_r\vee t_r>noc.ub_r)$. This means our confidence for the true selectivity being within the current cube bounds is lower bounded by $1-2\cdot\exp(-2\cdot d_r^2\cdot l_r)$. 

The confidence that the true selectivity falls within cube $noc$ is therefore given by \[\prod_{r\in pred(q)}\max(1-2\cdot\exp(-2\cdot d_r^2\cdot l_r),0).\]

\subsection{Sampling Schemes}
\label{samplingSub}

Function~\textproc{ChooseSamples}, used in Algorithm~\ref{paoAlg}, determines for each predicate how many samples to take when the next batch of samples is collected. We discuss different sampling strategies in this subsection. We do not provide pseudo-code but an implementation based on our descriptions is straight-forward.

We can split the problem of choosing the number of samples per predicate in the current iteration into two sub-problems. First, we must decide how many samples to take in total. Second, we must decide how to distribute those samples over the different predicates.

There are multiple metrics according to which we can compare sampling strategies. We can compare them in terms of the number of main loop iterations that they lead to but we can also compare them in terms of how many samples are taken in total or how often the standard optimizer is invoked. The number of optimizer invocations and the number of iterations are certainly correlated as we execute a bounded number of optimizer invocations in each iteration. The number of samples and the number of iterations are however not necessarily correlated and the right answer to the first question (how many samples to take in total) depends on the tradeoff that we want to realize between them. If we take a relatively large number of samples in each iteration then the number of required iterations and optimizer invocations will be reduced while the total number of samples might be higher than optimal. If we take a small number of samples then the number of iterations required to reach the PAO guarantees tends to increase. Note that the time required for sampling might depend more on the number of iterations than on the total number of samples if we can exploit parallelism on the sampling infrastructure (e.g., if data is processed on a big cluster or on a crowdsourcing platform such as Amazon Mechanical Turk\footnote{\url{https://www.mturk.com}}).

The optimal answer to the second question (how to distribute the total number of samples over different query predicates) depends on the samples collected so far, on the plan that is optimal for the current selectivity means, and on the extent of its near-optimality cube. Assume for instance that the current plan is near-optimal for all possible selectivity values of one predicate. Then sampling this predicate is not interesting. Or assume that the cost function of the current plan is equally sensitive to two predicates but that we already have very tight confidence bounds for one of them while the selectivity of the second one is rather uncertain. Clearly it is more interesting to sample the second predicate in this situation. The sampling strategies that we discuss in the following all consider a subset of the issues we raised before and we compare different sampling strategies formally in Section~\ref{analysisSec} and experimentally in Section~\ref{experimentalSec}.

The simplest sampling strategy that we consider is uniform sampling. We take the same number of samples in each iteration and samples are distributed uniformly over the predicates.

A slightly more sophisticated sampling strategy is exponential sampling. The total number of samples is still distributed uniformly over the different predicates. However, the total number of samples to take is chosen according to the following principle. We take the same number of samples as uniform sampling in the first iteration. Then, denoting by $h$ the total number of samples taken in all prior iterations, we take $\rho\cdot h$ samples in the current iteration. Factor $\rho>0$ is a user-defined constant that allows to trade between the number of iterations and the number of samples taken. We will see during formal analysis and during our experiments that this seemingly simple sampling strategy often achieves a good tradeoff between the number of iterations, samples, and optimizer invocations.

All sampling strategies discussed so far distribute the total number of samples uniformly over all predicates. This is not the case for our last sampling strategy that we call adaptive sampling. Adaptive sampling adapts the number of samples to take for each predicate to the plan that is optimal for the current selectivity means and to its near-optimality cube. We make the conservative assumption that new samples will not change the current selectivity means significantly, that the plan that is optimal for the current selectivity means will remain optimal in the next iteration, and that its near-optimality cube remains the same as well. Under those assumptions, we can calculate the number of additional samples required for each predicate that are required in order to reach the required confidence value $\delta$. More precisely, denote by $m$ the number of query predicates. Then the confidence value $\delta$ is reached once the confidence for each single predicate that its true selectivity value falls within the cube bounds reaches value $\delta^{1/m}$. For each predicate, we must therefore set $1-2\cdot\exp(-2\cdot d_r^2\cdot l_r)\geq\delta^{1/m}$ which allows to calculate the required number of samples.

%% file: sections/analysis.tex
\section{Formal Analysis}
\label{analysisSec}

We analyze the maximal number of samples that the algorithm presented in the last section has to take in order to satisfy the probabilistic PAO guarantees. 

The number of required samples depends on properties of the cost function. We assume in the following that multiplying the selectivity estimate for each of $m$ predicates by factor $f$ might multiply the cost of a plan by up to factor $f^m$. We call this assumption the maximal sensitivity assumption (MSA) in the following, as cost functions are assumed to be equally sensitive to a change in any of the predicate selectivity values. We require MSA only to hold locally in a small area around the true selectivity values. Note that our algorithm works even if MSA is not satisfied; the assumption is uniquely required for our formal analysis.

MSA is a pessimistic assumption as having cost functions that are very sensitive to the selectivity estimates means that plans tend to be near-optimal within smaller regions of the selectivity space. Then the number of samples required to shrink the selectivity confidence bounds until they are contained within those near-optimality regions increases. 

We first establish a minimal extent for near-optimality regions under the MSA. Then we calculate the number of samples required to shrink the confidence bounds. We denote in the following by $m$ the number of query predicates.

\begin{lemma}
If a plan is optimal at selectivity vector $s$ then it is $\alpha$-optimal in the cube bounded by $s\alpha^{-1/(2m)}$ and $s\alpha^{1/(2m)}$.\label{minCubeLemma}
\end{lemma}
\begin{proof}
Plan $p$ is in the following the plan that is optimal for selectivity vector $s$ and $c$ the cost of $p$ at $s$. We first calculate a lower cost bound for the entire cube. A different plan than $p$ might be optimal at the lower cube bounds $s\alpha^{-1/(2m)}$. Let $p_L$ denote the plan that is optimal for the lower cube bounds. The cost of $p_L$ for $s$ is lower-bounded by $c$ as $c$ is the optimal cost at $s$. Multiplying the selectivity by factor $\alpha^{-1/(2m)}$ can reduce the plan cost at most by factor $\alpha^{-1/2}$ due to the MSA. Therefore $c\cdot \alpha^{-1/2}$ is a lower bound for the cost at the lower cube bound which lower-bounds at the same time the cost in the entire cube as we assume that cost functions are monotone in the selectivity values. The cost of $p$ at the upper cube bound is an upper bound on the cost of $p$ within the entire cube. The cost of $p$ at $s\alpha^{1/(2m)}$ is upper-bounded by $c\cdot\alpha^{1/2}$ due to MSA. It is $c\alpha^{1/2}\leq c\alpha^{-1/2}\cdot\alpha$ and the sufficient condition for having a near-optimality cube (see Theorem~\ref{nearOptTheorem} in Section~\ref{cubeSub}) satisfied.
\end{proof}

The confidence that the true selectivity falls within the cube bounds is for each predicate calculated based on the minimal distance between mean and cube bounds.

 \begin{lemma}
 The plan $p$ that is optimal for the current selectivity means $a$ is $\alpha$-optimal in a cube around $a$ with minimal distance $a\cdot(1-\alpha^{-1/(2m)})$ between $a$ and the cube bounds.
 \end{lemma}
 \begin{proof}
Lemma~\ref{minCubeLemma} specifies upper and lower bounds of a cube in which $p$ is $\alpha$-optimal. The distance between $a$ and the upper cube bound is $a\cdot(\alpha^{1/(2m)}-1)$ and the distance between $a$ and the lower cube bound is $a\cdot(1-\alpha^{-1/(2m)})$. It is $1-1/\alpha^{1/(2m)}=(\alpha^{1/(2m)}-1)/\alpha^{1/(2m)}\leq \alpha^{1/(2m)}-1$ for $\alpha\geq1$. The minimal distance between $a$ and the cube bounds is therefore the distance to the lower cube bounds.
 \end{proof}
 
We calculate the number of samples required to achieve confidence $\delta$. The following analysis is pessimistic since we assume that plans are only near-optimal within the minimal near-optimality cube described in Lemma~\ref{minCubeLemma}. 

We assume that the cube extension algorithm is able to establish that the best plan for the current selectivity means is near-optimal at least within its minimal cube. This seems particularly likely for the slope-based extension scheme (Algorithm~\ref{cubeComplexAlg}): the target cost bounds it aims to achieve for the upper and lower cube bounds correspond precisely to the ones of the minimal selectivity cube if the plan cost functions are as sensitive to selectivity changes as maximally allowed by the MSA and if the plan that is optimal for the selectivity means is also optimal at the lower cube bounds.

We simplify the following analysis by assuming that all predicates have the same selectivity mean $a$ and therefore the same minimal distance $d=a\cdot(1-\alpha^{-1/(2m)})$ between $a$ and the cube bounds, while the generalization is straight forward. We also assume that the selectivity means $a$ remain fixed over the iterations (alternatively, we can say that $a$ denotes the selectivity means in the final iteration without affecting the following analysis). We currently assume a uniform distribution of samples over predicates.

\begin{theorem}
An $\alpha$-optimal plan is found with confidence $\delta$ after taking $m\cdot\log(2/(1-\sqrt[m]{\delta}))/(2\cdot d^2)$ samples.
\end{theorem}
\begin{proof}
The algorithm terminates once the confidence that the true selectivity values lie within the minimal near-optimality cube reaches the threshold $\delta$. This is the case once the confidence for each of the $m$ predicates that its selectivity lies within the corresponding cube bounds reaches value $\sqrt[m]{\delta}$. Hence we require $1-2\exp(-2\cdot d^2\cdot l)\geq\sqrt[m]{\delta}$ where $l$ denotes the number of samples taken for that predicate. This leads to the condition $\exp(-2\cdot d^2\cdot l)\leq(1-\sqrt[m]{\delta})/2$ and $-2\cdot d^2\cdot l\leq\log((1-\sqrt[m]{\delta})/2)$ and finally $l\geq\log(2/(1-\sqrt[m]{\delta}))/(2\cdot d^2)$. This formula specifies the required samples per predicate so we need to multiply by $m$ to obtain the total number of required samples.
\end{proof}

The required number of samples grows in the confidence threshold $\delta$. It grows once the minimal distance between mean and cube bounds decreases; the minimal distance increases in $\alpha$. As expected, requiring a higher confidence or using a stricter definition of near-optimality increases the required number of samples. Our formula for the number of samples also grows for lower selectivity values (as the minimal distance is monotone in the selectivity). Note however that our formula establishes an upper bound on the required number of samples under pessimistic assumptions. The actual number of required samples might be less sensitive to low selectivity values and we will show in our experiments that our algorithms perform well for relatively low selectivity values of few per mille in average. More importantly, low selectivity values mean low execution cost as the size of intermediate results becomes very small. We discussed the natural extension of our algorithm to terminate iterations once the estimated execution cost becomes too low (e.g., with  confidence $\delta$). Low selectivity values lead to low execution cost and would trigger early termination. 

Also, there are specialized sampling methods that reduce sampling effort for predicates with low selectivity in specific domains. A counting based selectivity estimation method has for instance been proposed recently in the context of crowd database systems~\cite{Marcus2012} that can reduce the required number of samples significantly, in particular for image recognition tasks.

Note that the number of samples does not at all depend on the size of the processed data! The number of required samples will therefore become more and more negligible the bigger the data set on which the query operates is.

We denote by $u(\alpha,\delta,m)$ the upper bound on the required number of samples in the following. We finally study how the required number of samples translates into the actual number of samples taken, the number of iterations and the number of optimizer calls, for the different sampling strategies described in Section~\ref{samplingSub}.

\begin{theorem}
The number of samples, optimizer calls, and iterations are all in $O(u(\alpha,\delta,m))$ for uniform sampling.
\end{theorem}
\begin{proof}
Uniform sampling takes a constant number of samples per iteration and the number of optimizer calls per iteration is bounded by a constant, too. 
\end{proof}

We denote the multiplicative factor used by the exponential sampling strategy by $\rho$ in the following.

\begin{theorem}
The number of iterations and optimizer calls are both in $O(\log_{\rho}u(\alpha,\delta,m))$ while the number of samples is in $O(\rho\cdot u(\alpha,\delta,m))$ for exponential sampling.
\end{theorem}
\begin{proof}
The total number of samples taken in each iteration equals the total number of samples taken so far multiplied by factor $\rho$. Hence the total number of samples increases by factor $\rho+1$ after each iteration. The number of iterations is in $O(\log_{\rho+1}u(\alpha,\delta,m))=O(\log_{\rho}u(\alpha,\delta,m))$ and the number of optimizer calls is proportional. The number of iterations determines the number of samples which is $(\rho+1)^{\lceil\log_{\rho+1}u(\alpha,\delta,m)\rceil}\in O(\rho\cdot u(\alpha,\delta,m))$.
\end{proof}

Exponential sampling reduces the asymptotic number of iterations and optimizer invocations significantly while it increases the number of samples by factor $\rho$. Choosing a higher value for $\rho$ decreases the number of iterations and optimizer calls while it may increase the number of samples. Note that our analysis is based on the simplifying assumptions that a plan is only near-optimal within its minimal cube and that this cube is reliably discovered in each iteration. We will see in our experiments that the number of samples increases sometimes by more than factor $\rho$ when switching from uniform to exponential sampling.

Adaptive sampling tries to calculate the required number of samples for each predicate based on the current selectivity estimates. This results in one iteration and a constant number of optimizer calls under the (in practice overly optimistic) assumption that the selectivity values do not change significantly enough to invalidate the forecast. Based on that property, we can however certainly expect a relatively low number of iterations and optimizer calls. In addition, the adaptive sampling scheme is the only one that does not distribute samples uniformly over different predicates. This means that the number of samples taken by the adaptive scheme is for instance guaranteed to be lower than the one taken by all other schemes at least by factor $m$ in cases where selectivity matters only for one out of the $m$ predicates (e.g., if all predicates but one are defined on small tables whose processing cost is negligible).

%% file: sections/experimental.tex
\section{Experimental Evaluation}
\label{experimentalSec}

We experimentally compare the different variants of the PAO algorithm described in Section~\ref{algorithmSec} against each other. 

Section~\ref{setupSub} describes the experimental setup. Section~\ref{extensionExpSub} presents the result of a comparison between different algorithms for finding selectivity space regions and Section~\ref{samplingExpSub} compares different predicate sampling strategies. 

\def\alphaLow{3}
\def\alphaHigh{5}
\def\deltaLow{0.5}
\def\deltaHigh{0.9}

\def\plotTitle#1#2#3{
\node at (group c#1r#2.north) [yshift=0.15cm] {#3};
}

\def\regionResultPlot#1#2#3#4{
\begin{groupplot}[group style={group size=2 by 2, x descriptions at=edge bottom, ylabels at=edge left, vertical sep=0.35cm}, width=3.5cm, height=2.7cm, xlabel near ticks, ylabel near ticks, xlabel=\# Predicates, ylabel=#1, ybar=0pt, ymode=log, ymajorgrids, xtick align=inside,
legend style={font=\footnotesize}]
\nextgroupplot[bar width=0.1cm]
\foreach \i in {2,3,4}{
\addplot table[x index=1, y index=\i] {plotsdata/regionPlots/#2#3#4.txt};
}
\nextgroupplot[bar width=0.1cm,legend pos=outer north east, legend columns=1]
\foreach \i in {8,9,10}{
\addplot table[x index=1, y index=\i] {plotsdata/regionPlots/#2#3#4.txt};
}
\addlegendentry{Slope}
\addlegendentry{FR($10^{-2}$)}
\addlegendentry{FR($10^{-3}$)}
\nextgroupplot[bar width=0.1cm]
\foreach \i in {14,15,16}{
\addplot table[x index=1, y index=\i] {plotsdata/regionPlots/#2#3#4.txt};
}
\nextgroupplot[bar width=0.1cm]
\foreach \i in {20,21,22}{
\addplot table[x index=1, y index=\i] {plotsdata/regionPlots/#2#3#4.txt};
}

\end{groupplot}

\plotTitle{1}{1}{{$\alpha=$\alphaLow, $\delta=$\deltaLow}}
\plotTitle{2}{1}{{$\alpha=$\alphaLow, $\delta=$\deltaHigh}}
\plotTitle{1}{2}{{$\alpha=$\alphaHigh, $\delta=$\deltaLow}}
\plotTitle{2}{2}{{$\alpha=$\alphaHigh, $\delta=$\deltaHigh}}
}

\def\samplingResultPlot#1#2#3#4{
\begin{groupplot}[group style={group size=2 by 2, x descriptions at=edge bottom, ylabels at=edge left, vertical sep=0.35cm}, width=3.5cm, height=2.7cm, xlabel near ticks, ylabel near ticks, xlabel=\# Predicates, ylabel=#1, ybar=0pt, ymode=log, ymajorgrids, xtick align=inside,
legend style={font=\footnotesize}]
\nextgroupplot[bar width=0.1cm]
\foreach \i in {2,3,4,5}{
\addplot table[x index=1, y index=\i] {plotsdata/samplingPlots/#2#3#4.txt};
}
\nextgroupplot[bar width=0.1cm, legend pos=outer north east, legend columns=1]
\foreach \i in {6,7,8,9}{
\addplot table[x index=1, y index=\i] {plotsdata/samplingPlots/#2#3#4.txt};
}
\addlegendentry{Uni}
\addlegendentry{Exp(0.5)}
\addlegendentry{Exp(1)}
\addlegendentry{Adapt}
\nextgroupplot[bar width=0.1cm]
\foreach \i in {10,11,12,13}{
\addplot table[x index=1, y index=\i] {plotsdata/samplingPlots/#2#3#4.txt};
}
\nextgroupplot[bar width=0.1cm]
\foreach \i in {14,15,16,17}{
\addplot table[x index=1, y index=\i] {plotsdata/samplingPlots/#2#3#4.txt};
}
\end{groupplot}

\plotTitle{1}{1}{{$\alpha=$\alphaLow, $\delta=$\deltaLow}}
\plotTitle{2}{1}{{$\alpha=$\alphaLow, $\delta=$\deltaHigh}}
\plotTitle{1}{2}{{$\alpha=$\alphaHigh, $\delta=$\deltaLow}}
\plotTitle{2}{2}{{$\alpha=$\alphaHigh, $\delta=$\deltaHigh}}
}

\def\regionFiguresPack#1#2#3{
\begin{figure}[t!]
\centering
\begin{tikzpicture}
\regionResultPlot{{Iterations}}{iterations}{#2}{#3}
\end{tikzpicture}

\caption{Number of iterations for #1 queries and different extension algorithms.\label{FigIter#2#3}}
\end{figure}

\begin{figure}[t!]
\centering
\begin{tikzpicture}
\regionResultPlot{{Calls}}{calls}{#2}{#3}
\end{tikzpicture}

\caption{Number of optimizer calls for #1 queries and different extension algorithms.\label{FigCall#2#3}}
\end{figure}}

\def\samplingFiguresPack#1#2#3#4{
\begin{figure}[t!]
\centering
\begin{tikzpicture}
\samplingResultPlot{{Iterations}}{iterations}{#2}{#3}
\end{tikzpicture}

\caption{Number of iterations for #1 queries and different sampling schemes.\label{FigSampleIter#2#3}}
\end{figure}

\begin{figure}[t!]
\centering
\begin{tikzpicture}
\samplingResultPlot{{Samples}}{samples}{#2}{#3}
\end{tikzpicture}

\caption{Number of samples for #1 queries and different sampling schemes.\label{FigSampleSample#2#3}}
\end{figure}

#4

\begin{figure}[t!]
\centering
\begin{tikzpicture}
\samplingResultPlot{{Calls}}{calls}{#2}{#3}
\end{tikzpicture}

\caption{Number of optimizer calls for #1 queries and different sampling schemes.\label{FigSampleCalls#2#3}}
\end{figure}
}

\subsection{Experimental Setup}
\label{setupSub}

We compare algorithms in terms of the number of iterations they perform, the number of standard optimizer calls that they require, and in terms of the total number of samples that they collect for solving a given PAO instance. We compare algorithms on randomly generated queries and use the query generation method described by Steinbrunn et al.~\cite{Steinbrunn1997} (which was reused in other publications~\cite{Bruno}). This generation method results in predicates with rather moderate selectivity (median selectivity over 100 generated queries joining 10 tables was around 3 per mille). After generating a query, we randomly pick a given number of predicates for which no selectivity values are provided to the optimizer. Our optimization algorithm can only approximate those selectivity values via sampling  while selectivity estimates for the other predicates are known. We simulate evaluating predicates on data samples by flipping a coin for each tuple with a success probability that corresponds to the selectivity that was determined at query generation. 

We generate queries with different numbers of joined tables and predicates. We performed one series of experiments on chain queries and one on star queries; the results show the same tendencies and hence we report only the results for chain queries in the following. We compare algorithms for different approximation factors $\alpha$ and for different confidence thresholds $\delta$. Selectivity estimates are initialized by sampling each predicate until the first tuple satisfying the predicate is encountered~\cite{Karanasos2014a}. The number of samples taken per iteration by uniform sampling corresponds to 0.25 percent of the average number of predicate evaluations required for the first join. Exponential sampling takes the same number of samples in the first iteration. The maximal number of samples per iteration is restricted to 100 times that number for adaptive sampling. The number of standard optimizer invocations per iteration is restricted to 10. 

All algorithms were implemented in Java~1.7 and executed on an iMac with  i5-3470S 2.90GHz CPU and 16~GB of RAM. 

\subsection{Comparing Cube Extension Algorithms}
\label{extensionExpSub}

\regionFiguresPack{5-table chain}{Small}{Chain}
\regionFiguresPack{10-table chain}{Big}{Chain}


We compare the two algorithms presented in Section~\ref{cubeSub} for finding near-optimality cubes using a given number of standard optimizer invocations. ``Slope'' designates the extension scheme that exploits the slope of the plan cost function (see Algorithm~\ref{cubeComplexAlg}) while ``FR($\epsilon$)'' designates the naive version (see Algorithm~\ref{cubeNaiveAlg}) that uses a fixed extension ratio defined by $\epsilon$. We experimented with different values for $\epsilon$ in preliminary experiments and report only results for the two best configurations in the following plots. All algorithms use the uniform sampling scheme in this subsection. 

Figures~\ref{FigIterSmallChain} to \ref{FigCallBigChain} compare the algorithms in terms of how many iterations are necessary (which is at the same time proportional to the number of required samples as we use uniform sampling) and how many standard optimizer calls are necessary. We compare algorithms for two query sizes and four different combinations of values for $\alpha$ and $\delta$ and for four different numbers of parameters with unknown selectivity values. Each data point represents the arithmetic average over 25 test cases. We interrupted iterations of the fixed ratio algorithm once the number of iterations exceeds the iterations required by Slope by more than one order of magnitude (this happened in average for one out of 25 test cases). Hence the reported numbers for FR correspond to lower bounds on the real values.


Slope wins according to the number of iterations (and samples) and according to the number of optimizer calls. The difference is significant (note the logarithmic axis) and shows the impact of the improvements that were discussed in Section~\ref{cubeSub}. FR is less efficient than Slope at approximating the near-optimality region of a given plan. This means that FR requires tighter confidence bounds for the selectivity values before it can terminate and this results in the increased number of iterations. 

The number of required iterations as well as the required number of optimizer calls increase for all algorithms once $\delta$ increases or $\alpha$ decreases as both corresponds to stronger probabilistic guarantees. Increasing the number of predicates increases the number of iterations and optimizer calls as well. Note however that Slope can easily optimize 10-table queries with eight predicates having unknown selectivity values, using only around 1100 standard optimizer calls even for the strongest probabilistic guarantees. Algorithms for parametric query optimization~\cite{D.2008, Bizarro2009}, where the goal is to find optimal plans for all possible selectivity value combinations, are for instance typically evaluated on lower numbers of selectivity parameters and incur higher computational overhead. Our approach for cube extension is tailored to our scenario as we only approximate the near-optimality region of one single plan instead of finding near-optimal plans for each region in the selectivity space. In the next subsection, we will see how the number of iterations can be reduced by more sophisticated sampling schemes.

\subsection{Comparing Sampling Schemes}
\label{samplingExpSub}

\samplingFiguresPack{5-table chain}{Small}{Chain}{}
\samplingFiguresPack{10-table chain}{Big}{Chain}{}

Figures~\ref{FigSampleIterSmallChain} to \ref{FigSampleCallsBigChain} compare the sampling strategies that were presented in Section~\ref{samplingSub} over different query sizes, predicate counts, and PAO guarantees. ``Uni'' denotes the uniform sampling strategy, ``Exp($\rho$)'' is the exponential sampling strategy with multiplicative factor $\rho$ and ``Adapt'' designates the adaptive sampling strategy in those plots. 

The comparison shows that different sampling strategies are optimal according to different cost metrics. Uniform sampling requires the lowest number of samples but the highest number of iterations and optimizer calls. It is preferable if sampling is rather expensive compared with optimization overhead and if the sampling infrastructure offers a low degree of parallelism. The other sampling schemes reduce the number of required iterations however often by one order of magnitude or more for elevated number of predicates and high confidence requirements. 

The exponential sampling strategy realizes different tradeoffs between the number of optimizer calls and iterations and the number of samples depending on the choice for $\rho$. Choosing a higher value for $\rho$ decreases the number of iterations and optimizer calls while it usually increases the number of samples. Comparing with the other sampling strategies, the exponential sampling strategies offer often a good tradeoff between number of optimizer invocations and samples. 

The adaptive sampling scheme usually requires the lowest number of iterations and optimizer calls among all sampling strategies. This shows that estimating how many samples are required to meet the confidence requirements is helpful even if the theoretical optimum of only one iteration is not reached. On the other hand, adaptive sampling has often the highest number of samples for moderate number of query predicates. Once the number of samples becomes elevated (at least four predicates for high confidence requirements and at least six for low confidence requirements), adaptive sampling starts requiring less samples than the exponential sampling strategies. This shows the benefit of distributing samples non-uniformly over predicates. 

Altogether we recommend uniform sampling to minimize the number of samples, the exponential strategies to achieve a good and tunable tradeoff between the number of iterations and samples for queries with moderate number of predicates, and the adaptive sampling strategy to minimize the number of iterations and for queries with many predicates.

%% file: sections/conclusion.tex
\section{Conclusion}
\label{conclusionSec}

We introduced probably approximately optimal query optimization and a first corresponding algorithm. 


%% file: main.bbl
\begin{thebibliography}{10}

\medskip

\bibitem{Abhirama2010a}
M.~Abhirama, S.~Bhaumik, and A.~Dey.
\newblock {On the stability of plan costs and the costs of plan stability}.
\newblock {\em VLDB}, 3(1):1137--1148, 2010.

\bibitem{Avnur2000}
R.~Avnur and J.~Hellerstein.
\newblock {Eddies: Continuously adaptive query processing}.
\newblock In {\em SIGMOD}, pages 261--272, 2000.

\bibitem{Babcock2005}
B.~Babcock and S.~Chaudhuri.
\newblock {Towards a robust query optimizer: a principled and practical
  approach}.
\newblock In {\em SIGMOD}, pages 119--130, 2005.

\bibitem{Babu2005}
S.~Babu, P.~Bizarro, and D.~DeWitt.
\newblock {Proactive Re-Optimization}.
\newblock In {\em SIGMOD}, pages 107--118, New York, New York, USA, 2005. ACM
  Press.

\bibitem{Bizarro2009}
P.~Bizarro, N.~Bruno, and D.~DeWitt.
\newblock {Progressive parametric query optimization}.
\newblock {\em KDE}, 21(4):582--594, 2009.

\bibitem{Bruno}
N.~Bruno.
\newblock {Polynomial heuristics for query optimization}.
\newblock In {\em ICDE}, pages 589--600, 2010.

\bibitem{Bruno2002}
N.~Bruno and S.~Chaudhuri.
\newblock {Exploiting Statistics on Query Expressions for Optimization}.
\newblock In {\em SIGMOD}, pages 263--274, 2002.

\bibitem{Chaudhuri2009a}
S.~Chaudhuri and B.~Babcock.
\newblock {Query selectivity estimation with confidence interval}, 2009.

\bibitem{Chaudhuri2001}
S.~Chaudhuri and V.~Narasayya.
\newblock {Automating statistics management for query optimizers}.
\newblock In {\em ICDE}, pages 7--20, 2001.

\bibitem{Cole1994}
R.~Cole and G.~Graefe.
\newblock {Optimization of dynamic query evaluation plans}.
\newblock In {\em SIGMOD}, pages 150--160, 1994.

\bibitem{D.2008}
H.~D., P.~N. Darera, and J.~R. Haritsa.
\newblock {Identifying robust plans through plan diagram reduction}.
\newblock In {\em VLDB}, pages 1124--1140, 2008.

\bibitem{Dey2008}
A.~Dey, S.~Bhaumik, and J.~Haritsa.
\newblock {Efficiently Approximating Query Optimizer Plan Diagrams}.
\newblock In {\em PVLDB}, pages 1325--1336, 2008.

\bibitem{Dutt2014a}
A.~Dutt.
\newblock {QUEST : An exploratory approach to robust query processing}.
\newblock {\em VLDB}, 7(13):5--8, 2014.

\bibitem{Dutt2014}
A.~Dutt and J.~Haritsa.
\newblock {Plan bouquets: query processing without selectivity estimation}.
\newblock In {\em SIGMOD}, pages 1039--1050, 2014.

\bibitem{Franklin2011}
M.~Franklin and D.~Kossmann.
\newblock {CrowdDB: answering queries with crowdsourcing}.
\newblock {\em SIGMOD}, 2011.

\bibitem{Ganguly1998}
S.~Ganguly.
\newblock {Design and Analysis of Parametric Query Optimization Algorithms}.
\newblock In {\em VLDB}, pages 228--238, 1998.

\bibitem{Garcia-molina2014}
H.~Garcia-molina.
\newblock {Exploiting Correlations for Expensive Predicate Evaluation}.
\newblock {\em VLDB}, 7(4), 2014.

\bibitem{Goel2006}
A.~Goel, A.~Goel, S.~Guha, S.~Guha, K.~Munagala, and K.~Munagala.
\newblock {Asking the right questions: model-driven optimization using probes}.
\newblock In {\em PODS}, pages 203--212, 2006.

\bibitem{Graefe1989}
G.~Graefe and K.~Ward.
\newblock {Dynamic query evaluation plans}.
\newblock In {\em SIGMOD}, pages 358--366, 1989.

\bibitem{Guha2007}
S.~Guha and K.~Munagala.
\newblock {Model-driven optimization using adaptive probes}.
\newblock In {\em ACM-SIAM symposium on Discrete algorithms}, pages 308--317,
  2007.

\bibitem{Hoeffding63}
W.~Hoeffding.
\newblock {Probability inequalities for sums of bounded random variables}.
\newblock {\em Journal of the American statistical association},
  58(301):13--30, 1963.

\bibitem{Hulgeri2002}
A.~Hulgeri and S.~Sudarshan.
\newblock {Parametric Query Optimization for Linear and Piecewise Linear Cost
  Functions}.
\newblock {\em VLDB}, 2002.

\bibitem{Karanasos2014a}
K.~Karanasos, A.~Balmin, M.~Kutsch, F.~Ozcan, V.~Ercegovac, C.~Xia, and
  J.~Jackson.
\newblock {Dynamically optimizing queries over large scale data platforms}.
\newblock In {\em SIGMOD}, pages 943--954, 2014.

\bibitem{Marcus2012}
A.~Marcus, D.~Karger, S.~Madden, R.~Miller, and S.~Oh.
\newblock {Counting with the crowd}.
\newblock {\em VLDB}, 6(2):109--120, 2012.

\bibitem{Marcus2011a}
A.~Marcus, E.~Wu, and D.~Karger.
\newblock {Human-powered sorts and joins}.
\newblock In {\em VLDB}, pages 13--24, 2011.

\bibitem{Marcus2012a}
A.~Marcus, E.~Wu, D.~R. Karger, S.~Madden, R.~C. Miller, S.~Acm, and N.~York.
\newblock {Demonstration of Qurk : A query processor for human operators}.
\newblock In {\em SIGMOD}, pages 1315--1318, 2011.

\bibitem{Park2013a}
H.~Park, R.~Pang, A.~Parameswaran, H.~Garcia-Molina, N.~Polyzotis, and
  J.~Widom.
\newblock {An overview of the deco system: data model and query language; query
  processing and optimization}.
\newblock {\em SIGMOD}, 41:22--27, 2013.

\bibitem{Park2013}
H.~Park and J.~Widom.
\newblock {Query optimization over crowdsourced data}.
\newblock {\em VLDB}, pages 781--792, 2013.

\bibitem{Seppi1993}
K.~D. Seppi, J.~W. Barnes, and C.~N. Morris.
\newblock {A Bayesian Approach to Database Query Optimization}.
\newblock {\em INFORMS Journal on Computing}, 5(4):410--419, 1993.

\bibitem{Steinbrunn1997}
M.~Steinbrunn, G.~Moerkotte, and A.~Kemper.
\newblock {Heuristic and randomized optimization for the join ordering
  problem}.
\newblock {\em VLDBJ}, 6(3):191--208, Aug. 1997.

\bibitem{Stillger2001}
M.~Stillger, G.~M. Lohman, V.~Markl, and M.~Kandil.
\newblock {LEO - DB2's LEarning Optimizer}.
\newblock In {\em VDDB}, pages 19--28, 2001.

\bibitem{Valiant1984}
L.~G. Valiant.
\newblock {A theory of the learnable}.
\newblock {\em Communications of the ACM}, 27:1134--1142, 1984.

\end{thebibliography}
